\documentclass[a4paper,UKenglish]{lipics}
\usepackage{amsmath, amssymb}

\usepackage{todonotes}
 
\usepackage{enumerate}

\usepackage{graphicx}
\usepackage{mathtools}
\usepackage{amssymb}

\usepackage{amsthm}

\usepackage{cite}
\usepackage{hyperref}

\newcommand{\myparagraph}[1]{\smallskip\noindent{\textbf{\sffamily #1} \ }}

\usepackage{algorithm2e}

\newtheorem{proposition}{Proposition}

\newtheorem{krule}{Reduction Rule}

\newcommand{\inn}[1]{{#1}_{\textsf{int}}}
\newcommand{\lp}{\frac{1-\lambda}{2}}

\newcommand{\ex}{ex} 
\newcommand{\mak}{\inf_{AK}}

\newcommand{\name}[1]{\textsc{#1}}
\newcommand{\maxcut}{\name{Max-Cut}}

\newcommand{\mcatlb}{\name{Max-Cut ATLB}}
\newcommand{\qcol}{\name{Max $q$-Colorable Subgraph}} 
\newcommand{\maxdag}{\name{Oriented Max Acyclic Digraph}}
\newcommand{\lamex}{$\lambda$-extendible}
\newcommand{\stronglamex}{strongly $\lambda$-extendible}
\newcommand{\PandT}{Poljak and Turz{\'{i}}k}
\newcommand{\AbovePT}{\name{Above Poljak-Turz\'{\i}k ($\Pi$)}}

\newcommand{\APT}{APT($\Pi$)}

\newcommand{\FPT}{\textsf{FPT}}
\newcommand{\EEbound}{Edwards-Erd\H{o}s bound}

\newcommand{\cO}{\mathcal{O}}

\newenvironment{parnamedefn}[4]{
\par\addvspace{0.4\baselineskip}\fbox{%
\begin{minipage}[t]{0.9\linewidth}%
\begin{tabular}{p{18mm}p{115mm}}
    \multicolumn{2}{l}{\name{#1}} \\
        \textsl{Input:} & {#2} \\ 
        \textsl{Parameter:} & {#3} \\ 
        \textsl{Question:} & {#4} \\
   \end{tabular}
\end{minipage}}\par\addvspace{0.4\baselineskip}
}

\usepackage{comment}

\newcommand{\LE}{\Pi}  

\newcommand{\pisub}{\APT}

\newcommand{\ms}{\beta}  
\newcommand{\pt}{\gamma} 

\newcommand{\comp}{{\cal C}}  
\newcommand{\union}{\bigcup}  

\title{Polynomial Kernels for $\lambda$-extendible Properties Parameterized Above the Poljak-Turz\'{i}k Bound}
\titlerunning{Polynomial Kernels for $\lambda$-extendible Properties} 

\author[1]{Robert Crowston }
\author[1]{Mark Jones}
\author[1]{Gabriele Muciaccia}
\author[2]{Geevarghese Philip}
\author[3]{Ashutosh Rai}
\author[3,4]{Saket Saurabh}

\affil[1]{Royal Holloway, University of London, UK. 
  \texttt{\{robert,markj,G.Muciaccia\}@cs.rhul.ac.uk}}

\affil[2]{Max-Planck-Institut f\"{u}r Informatik (MPII), Germany. 
  \texttt{gphilip@mpi-inf.mpg.de}}
\affil[3]{Institute of Mathematical Sciences, India. 
  \texttt{\{ashutosh,saket\}@imsc.res.in}}
\affil[4]{University of Bergen, Norway.}

\authorrunning{Crowston, Jones, Muciaccia, Philip, Rai, Saurabh} 

\Copyright{Crowston et. al.}


\begin{document}

\maketitle

\begin{abstract}
  \PandT{} (\emph{Discrete Mathematics} 1986) introduced the
  notion of \lamex{} properties of graphs as a generalization of
  the property of being bipartite. They showed that for any
  $0<\lambda<1$ and \lamex{} property $\Pi$, any connected graph
  $G$ on $n$ vertices and $m$ edges contains a spanning subgraph
  $H\in\Pi$ with at least $\lambda{}m+\frac{1-\lambda}{2}(n-1)$
  edges. The property of being bipartite is \lamex{} for
  $\lambda=1/2$, and so the Poljak-Turz\'{i}k bound generalizes the
  well-known Edwards-Erd\H{o}s bound for \maxcut{}. Other examples
  of \lamex{} properties include: being an acyclic oriented graph,
  a balanced signed graph, or a $q$-colorable graph for some
  \(q\in\mathbb{N}\).
  
  Mnich et. al. (\emph{FSTTCS} 2012) defined the closely related
  notion of \emph{strong} $\lambda$-extendibility.  They showed
  that the problem of finding a subgraph satisfying a given
  strongly $\lambda$-extendible property $\Pi$ is fixed-parameter
  tractable (FPT{}) when parameterized above the Poljak-Turz\'{i}k
  bound---\emph{does there exist a spanning subgraph $H$ of a
    connected graph $G$ such that $H\in\Pi$ and $H$ has at least
    $\lambda{}m+\frac{1-\lambda}{2}(n-1)+k$ edges?}---subject to
  the condition that the problem is \FPT{} on a certain simple
  class of graphs called \emph{almost-forests of cliques}. This
  generalized an earlier result of Crowston et al. (\emph{ICALP}
  2012) for \maxcut{}, to all strongly \lamex{} properties which
  satisfy the additional criterion.

  In this paper we  settle the kernelization complexity
  of nearly all problems parameterized above Poljak-Turz\'{i}k
  bounds, in the affirmative.  We show that these problems admit
  quadratic kernels (cubic when $\lambda=1/2$), \emph{without
    using} the assumption that the problem is FPT on
  almost-forests of cliques. Thus our results not only remove the
  technical condition of being FPT on almost-forests of cliques
  from previous results, but also unify and extend previously
  known kernelization results in this direction. Our results add
  to the select list of \emph{generic} kernelization results known
  in the literature.


\end{abstract}

\section{Introduction}


In parameterized complexity each problem instance $I$ comes with a
parameter $k$, and a parameterized problem is said to be
\emph{fixed parameter tractable} (\FPT{}) if for each instance
$(I,k)$ the problem can be solved in time $f(k)|I|^{\cO{(1)}}$
where $f$ is some computable function. The parameterized problem
is said to admit a {\it polynomial kernel} if there is a
polynomial time algorithm, called a {\em kernelization} algorithm,
that reduces the input instance down to an instance with size
bounded by a polynomial $p(k)$ in $k,$ while preserving the
answer. This reduced instance is called a {\em $p(k)$ kernel} for
the problem. 
The study of kernelization is a major research frontier of
Parameterized Complexity; many important recent advances in the
area pertain to kernelization. These include general results
showing that certain classes of parameterized problems have
polynomial
kernels~\cite{Alon:2010vp,H.Bodlaender:2009ng,FominLST10,FominLMS12}
and randomized kernelization based on matroid
tools~\cite{KratschW12,KratschW12Soda}.  The recent development of
a framework for ruling out polynomial kernels under certain
complexity-theoretic
assumptions~\cite{BodlaenderDFH09,Dell:2010sh,FS08} has added a
new dimension to the field and strengthened its connections to
classical complexity.  For overviews of kernelization we refer to
surveys~\cite{Bodlaender09,GN07SIGACT} and to the corresponding
chapters in books on Parameterized Complexity
\cite{FlumGroheBook,Niedermeierbook06}.  In this paper we give a
generic kernelization result for a class of problems parameterized
above guaranteed lower bounds.

\myparagraph{Context and Related Work.} Many interesting graph
problems are about finding a largest subgraph \(H\) of the input
graph \(G\), where graph \(H\) satisfies some specified property
and its size is defined as the number of its edges. For many
properties this problem is \textsf{NP}-hard, and for some of these
we know nontrivial lower bounds for the size of \(H\). In these
latter cases, the apposite parameterization ``by problem size''
is: Given graph \(G\) and parameter \(k\in\mathbb{N}\), does \(G\)
have a subgraph \(H\) which has (i) the specified property and
(ii) at least \(k\) \emph{more} edges than the best known lower
bound? \maxcut{} is a sterling example of such a problem. The
problem asks for a largest \emph{bipartite} subgraph \(H\) of the
input graph \(G\); it is \textsf{NP}-complete~\cite{Karp1972}, and
the well-known \emph{\EEbound{}}~\cite{Edwards1973,Edwards1975}
tells us that any connected loop-less graph on $n$ vertices and
$m$ edges has a bipartite subgraph with at least
$\frac{m}{2}+\frac{n-1}{4}$ edges. This lower bound is also the
best possible, in the sense that it is tight for an infinite
family of graphs---for example, for the set of all cliques with an
odd number of vertices.

\PandT{} investigated the \emph{reason} why bipartite subgraphs
satisfy the \EEbound{}, and they abstracted out a sufficient
condition for \emph{any} graph property to have such a lower
bound. They defined the notion of a \lamex{} property for
\(0<\lambda<1\), and showed that for any \lamex{} property $\Pi$,
any connected graph $G=(V,E)$ contains a spanning subgraph
$H=(V,F)\in\Pi$ with at least
\(\lambda{}|E|+\frac{1-\lambda}{2}(|V|-1)\)
edges~\cite{PoljakTurzik1986}. The property of being bipartite is
\lamex{} for $\lambda = 1/2$, and so the \PandT{} result implies
the \EEbound{}.  Other examples of \lamex{} properties---with
different values of $\lambda$---include $q$-colorability and
acyclicity in oriented graphs.

In their pioneering paper which introduced the notion of
``above-guarantee'' parameterization, Mahajan and
Raman~\cite{MahajanRaman1999} posed the parameterized tractability
of \maxcut{} above its tight lower bound (\mcatlb{})---\emph{Given
  a connected graph \(G\) with \(n\) vertices and \(m\) edges and
  a parameter \(k\in\mathbb{N}\), does \(G\) have a bipartite
  subgraph with at least $\frac{m}{2}+\frac{n-1}{4}+k$
  edges?}---as an open problem. This was recently resolved by
Crowston et al. who showed that \mcatlb{} can be solved in
\(2^{\cO(k)}\cdot{}n^{4}\) time and has a kernel with
\(\cO(k^{5})\) vertices~\cite{CrowstonJonesMnich2012}. Following
this, Mnich et al.~\cite{MnichPhilipSaurabhSuchy2012} generalized
the \FPT{} result of Crowston et al. to \emph{all} graph
properties which (i) satisfy a (potentially) stronger notion which
they dubbed \emph{strong \(\lambda\)-extendibility}, and (ii) are
\FPT{} on a certain simple class of graphs called
\emph{almost-forests of cliques}. That is, they showed that for
any \stronglamex{} graph property \(\Pi\) which satisfies the
simplicity criterion, the following problem---called \AbovePT{},
or \APT{} for short---is \FPT{}: \emph{Given a connected graph
  \(G\) with \(n\) vertices and \(m\) edges and a parameter
  \(k\in\mathbb{N}\), does \(G\) have a spanning subgraph
  \(H\in\Pi\) with at least
  \(\lambda{}m+\frac{1-\lambda}{2}(n-1)+k\) edges?} Problems which
satisfy these conditions include \maxcut{}, \maxdag{}, \qcol{}
and, more generally, any graph property which is equivalent to
having a homomorphism to a fixed vertex-transitive
graph~\cite{MnichPhilipSaurabhSuchy2012}.

\myparagraph{Our Results and their Implications.} Our main result
is that for almost all \stronglamex{} properties \(\Pi\) of
(possibly oriented or edge-labelled) graphs, the \AbovePT{}
problem has kernels with \(\cO(k^{2})\) or \(\cO(k^{3})\) vertices. Here ``almost
all'' includes the following: (i) \emph{all} \stronglamex{}
properties for \(\lambda\neq{}\frac{1}{2}\), (ii) \emph{all}
\stronglamex{} properties which contain all orientations and
labels (if applicable) of the graph \(K_{3}\) (triangle), and (iii) all
\emph{hereditary} \stronglamex{} properties  for
simple or oriented graphs.
In particular, our result
implies kernels with \(\cO(k^{2})\) vertices for \qcol{} and other problems defined by
homomorphisms to vertex-transitive graphs.

We address both the questions left open by Mnich et
al.~\cite{MnichPhilipSaurabhSuchy2012}, albeit in different
ways. Firstly, we resolve the kernelization question for
\stronglamex{} properties, except for the special cases of
non-hereditary \(\frac{1}{2}\)-extendible properties which do not
contain some orientation or labelling of the triangle, or hereditary \(\frac{1}{2}\)-extendible properties
which do not contain some labelling of the triangle.
Note that for non-hereditary properties, we may expect to find kernelization very difficult,
as a large subgraph with the property can disappear entirely if we delete even a small
part of the graph. For the cases when the membership of the triangle depends
on its labelling, we may expect the rules of kernelization to
depend greatly on the family of labellings, and so it is 
difficult to produce a general result.

 Secondly,
we get rid of the simplicity criterion required by Mnich et
al. Showing that a specific problem is \FPT{} on almost-forests of
cliques takes---in general---a non-trivial amount of work, as can
be seen from the corresponding proofs for
\maxcut{}~\cite[Lemma~9]{CrowstonJonesMnich2011}, \maxdag{}, and
having a homomorphism to a vertex transitive
graph~\cite[Lemmas~27,~31]{MnichPhilipSaurabhSuchy2012arXiv}. Mnich
et al. had proposed that a way to get around this problem was to
find a logic which captures all problems which are \FPT{} on
almost-forests of cliques, and had left open the problem of
finding the right logic. The proof of our main result shows that
\emph{all} \stronglamex{} properties---save for the special
cases ---are \FPT{} on almost-forests of cliques: in fact, that they
have polynomial size kernels on this class of graphs. No
special logic is required to capture these problems, and this
answers their second open problem.

Formally, our main result is as follows:
\begin{theorem}\label{thm:main}
  Let \(0<\lambda<1\), and let \(\Pi\) be a \stronglamex{}
  property of (possibly oriented and/or labelled) graphs. Then the
  \AbovePT{} problem has a kernel on \(\cO(k^{2})\) vertices if
  conditions \ref{cond:neqhalf} or \ref{cond:triangles} holds, and a kernel on \(\cO(k^{3})\) vertices if only \ref{cond:hereditary} holds:
  \begin{enumerate}
    \item\label{cond:neqhalf} \(\lambda\neq{}\frac{1}{2}\);
    \item\label{cond:triangles} All orientations and labels (if applicable) of the graph
      \(K_{3}\) belong to \(\Pi\); 
    \item\label{cond:hereditary} \(\Pi\) is a hereditary property of simple or oriented graphs.
  \end{enumerate}
\end{theorem}

As a corollary, we get that a number of specific problems have
polynomial kernels when parameterized above their respective
Poljak-Turz\'{i}k bounds:
\begin{corollary}\label{cor:examples}
  The \AbovePT{} parameterization of \qcol{}, $q>2$, has a
  kernel on \(\cO(k^{2})\) vertices, and 
  the \AbovePT{} parameterization of \maxdag{} has a
  kernel on \(\cO(k^{3})\) vertices. Furthermore, the \AbovePT{} parameterization
of any problem which is defined by homomorphism to a vertex-transitive
graph with at least 3 vertices has a kernel on \(\cO(k^{2})\) vertices.
\end{corollary}

The corollary follows from Theorem~\ref{thm:main} using the fact that each of
these problems is \(\lambda\)-extendible for different values of
\(\lambda\)~\cite{MnichPhilipSaurabhSuchy2012}.

\myparagraph{An outline of the proof.} We now give an intuitive
outline of our proof of Theorem~\ref{thm:main}. 
Our proof starts from a key result of Mnich et al.
\begin{proposition}[\cite{MnichPhilipSaurabhSuchy2012}] \label{thm:GminusS}
  Let $\Pi$ be a strongly $\lambda$-extendible property and let
  $(G,k)$ be an instance of \APT{}. Then in polynomial time, we
  can either decide that $(G,k)$ is a {\sc Yes}-instance or find a
  set $S \subseteq V(G)$ such that $|S| < \frac{6k}{1-\lambda}$
  and $G-S$ is a forest of cliques.
\end{proposition}

\autoref{thm:GminusS} is a classical WIN/WIN result, and either
outputs that the given instance is a {\sc YES} instance or outputs
a set $S\subseteq{}V(G);|S|<\frac{6k}{1-\lambda}$. In the former
case we return a trivial \textsc{YES} instance. In the latter case
we know that $G-S$ is a forest of cliques and
$|S|<\frac{6k}{1-\lambda}$; thus $G-S$ has a very special
structure. For $\lambda \neq \frac{1}{2}$, or when all orientations or labels of
the graph $K_3$ have the property,
we show combinatorially
that if the combined sizes of the cliques are too big then either we can get some
``extra edges'', or we can apply a reduction rule. We then show
that the reduced instance has size polynomial in $k$. For
$\lambda=\frac{1}{2}$, we need the extra technical condition that
the property is hereditary,
and defined only for simple or oriented graphs.
In this case we can show that the problem either contains (all
orientations of) $K_3$, or is exactly \maxcut{}, or that we can
bound the number and sizes of the cliques. In any of these cases
the problem admits a polynomial kernel.

A block of a graph $G$ is a maximal $2$-connected subgraph of $G$. Note
that a block ${\cal B}$ of $G$ may consist of a single vertex and no
edges, if that vertex is isolated in $G$.

Let \(G, S\) be as in \autoref{thm:GminusS}, and let $Q$ be the
set of cut vertices of $G-S$. For any block $B$ of $G-S$, let
$\inn{B}=V(B)\setminus Q$ be the {\em interior} of $B$.  Let
$\mathcal{B}$ be the set of blocks of $G-S$. A {\em block
  neighbor} of a block $B$ is a block $B'$ such that $|V(B)\cap
V(B')|=1$. Given a sequence of blocks $B_0,B_1,\dots,B_l,B_{l+1}$
in $G-S$, the subgraph induced by $V(B_1)\cup\dots\cup V(B_l)$ is
a {\em block path} if, for every $1\le i\le l$, $V(B_i)$ contains
exactly two vertices from $Q$, and $B_i$ has exactly two block
neighbors $B_{i-1}$ and $B_{i+1}$. A block $B$ in $G-S$ is a {\em
  leaf block} if $V(B)$ contains exactly one vertex from $Q$. A
block in $G-S$ is an {\em isolated block} if it contains no vertex
from \(Q\). Observe that an isolated block has no block neighbour,
while a leaf block has \emph{at least} one block neighbour.

Let $\mathcal{B}_0$ and $\mathcal{B}_1$ be the set of isolated
blocks and leaf blocks, respectively, contained in
$\mathcal{B}$. Let $\mathcal{B}_2$ be the set of blocks
$B\in\mathcal{B}$ such that \(B\) is a block in some block path of
\(G-S\). Finally, let $\mathcal{B}_{\geq
  3}=\mathcal{B}\setminus(\mathcal{B}_0\cup\mathcal{B}_1\cup\mathcal{B}_2)$. Thus:
\begin{itemize} 
\item \(\mathcal{B}_{0}\) is the set of all blocks of \(G-S\)
  which contain no cut vertex of \(G-S\), and therefore have no
  block neighbour;
\item \(\mathcal{B}_{1}\) is the set of all blocks of \(G-S\)
  which contain exactly one cut vertex of \(G-S\), and therefore
  have at least one block neighbour;
\item \(\mathcal{B}_{2}\) is the set of all blocks of \(G-S\)
  which (i) contain exactly two cut vertices of \(G-S\), \emph
  {and} (ii) have exactly two block neighbours; and,
\item \(\mathcal{B}_{3}\) is the set of all the remaining blocks
  of \(G-S\). A block of \(G-S\) is in \(\mathcal{B}_{3}\) if and
  only if it (i) contains at least two cut vertices of \(G-S\),
  \emph{and} (ii) has at least three block neighbours.
\end{itemize}

In order to bound the number of vertices in $G-S$ it is enough to
bound (i) the number of blocks, and (ii) the size of each block.
When $\lambda \neq \frac{1}{2}$ or the property includes all
orientations and labellings of $K_3$, we show (Lemma~\ref{lem:positivecliques}) that all blocks with two or
more vertices have positive excess. Using this fact, we can bound
the number of vertices in blocks of $\mathcal{B}_1$ or
$\mathcal{B}_2$ directly, and it remains only to bound
$|\mathcal{B}_0|$.  In the remaining case, we have to bound each
of $|\mathcal{B}_0|, |\mathcal{B}_1|, |\mathcal{B}_2|,
|\mathcal{B}_{\geq 3}|$ and the size of each block separately. We
bound these numbers over a number of lemmas.




%

\section{Definitions}
We use \(\uplus\) to denote the disjoint union of sets. We use
``graph'' to denote simple graphs without self-loops, directions,
or labels, and use standard graph terminology used by
Diestel~\cite{Diestel} for the terms which we do not explicitly
define.  Each edge in an \emph{oriented} graph has one of two
directions \(\{<,>\}\), while each edge in a \emph{labelled} graph
has an associated label \(\ell\in{}L\) chosen from a finite set
\(L\).  A \emph{graph property} is a subclass of the class of all
(possibly labelled and/or oriented) graphs. For a labelled and/or
oriented graph $G$, we use $U(G)$ to denote the underlying simple
graph; for any graph property of simple graphs, we say that $G$
has the property if $U(G)$ does: for instance, $G$ is connected if
$U(G)$ is. For a (possibly labelled and/or oriented) graph
\(G=(V,E)\) and weight function \(w:E(G)\to\mathbb{R}^{+}\), we
use \(w(F)\) to denote the sum of the weights of all the edges in
\(F\subseteq{}E\). We use $K_j$ to denote the complete simple
graph on $j$ vertices for \(j\in\mathbb{N}\), and $K$ to denote an
arbitrary complete simple graph. For a graph property \(\Pi\), we
say that $K_j \in \LE$ if $G \in \LE$ for every (oriented,
labelled) graph $G$ such that \(U(G)=K_{j}\).  A connected
(possibly labelled and/or oriented) graph is a {\em tree of
  cliques} if the vertex set of each block of the graph forms a
clique. We use \(\comp(G)\) to denote the set of connected
components of graph \(G\).  A {\em forest of cliques} is a graph
whose connected components are trees of cliques.  A graph \(G\) is
$2$-connected if and only if it does not contain cut vertices.


Mnich et al.~\cite{MnichPhilipSaurabhSuchy2012} defined the
following variant of \PandT's notion of
\(\lambda\)-extendibility~\cite{PoljakTurzik1986}.

\begin{definition}\label{def:stronglambda}
  Let $\mathcal{G}$ be a class of (possibly labelled and/or
  oriented) graphs and let $0<\lambda<1$. A graph property $\LE$
  is \emph{strongly $\lambda$-extendible} on $\mathcal{G}$ if it
  satisfies the following properties:
\begin{itemize}
\item {{\sc Inclusiveness}
    $\{G\in\mathcal{G}:U(G)\in\{K_1,K_2\}\}\subseteq\Pi$. That is,
    \(K_{1}\in\Pi\), and every possible orientation and labelling
    of the graph \(K_{2}\) is in \(\Pi\);}
\item {{\sc Block additivity} $G\in\mathcal{G}$ belongs to $\Pi$
    if and only if every block of $G$ belongs to $\Pi$;}
\item {{\sc Strong $\lambda$-subgraph extension} Let
    $G\in\mathcal{G}$ and let $(U,W)$ be a partition of $V(G)$,
    such that $G[U]\in\Pi$ and $G[W]\in\Pi$. For any weight
    function \(w:E(G)\to\mathbb{R}^{+}\) there exists an
    $F\subseteq E(U,W)$ with $w(F)\geq\lambda w(E(U,W))$, such
    that $G-(E(U,W)\setminus F)\in\Pi$.}
\end{itemize}
\end{definition}
In the rest of the paper we use $\mathcal{G}$ to denote a class of
(possibly labelled and/or oriented) graphs, and $\LE$ to denote an
arbitrary---but fixed---\stronglamex{} property defined on
$\mathcal{G}$ for some \(0<\lambda<1\).  The focus of our work is
the following ``above-guarantee'' parameterized problem:

\begin{parnamedefn}%
  {\AbovePT{} (\APT{})}%
  {A connected graph $G=(V,E)$ and an integer $k$.}%
  {$k$}%
  {Is there a spanning subgraph $H=(V,F)\in\Pi$ of $G$\newline
    such that $|F| \geq \lambda |E| +
    \frac{1-\lambda}{2}(|V|-1)+k$?}%
\end{parnamedefn}

Let $G\in\mathcal{G}$. A {\em $\Pi$-subgraph} of $G$ is a
\textbf{spanning} subgraph of $G$ which is in $\Pi$. Let
$\ms_\LE(G)$ denote the maximum number of edges in any
$\Pi$-subgraph of $G$, and let $\pt_\LE(G)$ denote the
Poljak-Turz\'{i}k bound on $G$; that is,
\(\pt_\LE(G)=\lambda|E(G)|+\lp(|V(G)|-|\mathcal{C}(G)|)\).  The
\emph{excess of $\LE$ on $G$}, denoted $\ex_\LE(G)$, is equal to
$\ms_\LE(G)-\pt_\LE(G)$. Thus, given a connected graph \(G\) and
\(k\in\mathbb{N}\) as inputs, the \APT{} problem asks whether
\(\ex_{\Pi}(G)\geq{}k\). We omit the subscript \(\Pi\) when it is
clear from the context.  We use $\ex(K_j)$ to denote the minimum
value of $\ex(G)$ for any (oriented, labelled) graph $G$ such that
$K_j=U(G)$. Thus, for example, if $\ex(K_3) = t$ then any graph
$G$ with underlying graph $K_3$ has a $\Pi$-subgraph with at least
$\pt(G) + t$ edges, regardless of orientations or labellings on
the edges of $G$. We say that a strongly $\lambda$-extendible
property {\em diverges on cliques} if there exists
$j\in\mathbb{N}$ such that $\ex(K_j)>\lp$.  We say that a simple
connected graph $\widetilde{K}$ is an {\em almost-clique} if there
exists $V'\subseteq V(\widetilde{K})$ with $|V'|\leq 1$ (possibly
$V'$ is empty) such that $\widetilde{K}-V'$ is a clique.  For an
almost-clique $\widetilde{K}$, we use $\ex(\widetilde{K})$ to
denote the minimum value of $\ex(G)$ for any (oriented, labelled)
graph $G$ such that $\widetilde{K}=U(G)$, and we say that
$\widetilde{K}\in\LE$ if and only if $G \in \LE$ for every
(oriented, labelled) graph $G$ with underlying graph
$\widetilde{K}$.

%

\begin{definition}
  We use $AK_{\Pi}^+$ to denote the class of all graphs
  $G\in\mathcal{G}$ such that $U(G)$ is an almost-clique and
  $\ex_{\Pi}(G)>0$. For any strongly $\lambda$-extendible property
  which diverges on cliques, we use $\mak$ to denote the value
  $\inf_{(G\in AK^+)}\ex(G)$.
\end{definition}

Note that the class $AK_{\Pi}^+$ contains an infinite number of graphs. Hence, it could be the case that $\mak=0$.
In the next section, we will show that for any strongly
$\lambda$-extendible property which diverges on cliques, it holds that $\mak>0$.

\section{Preliminary Results}

We begin with some preliminary results. The first two lemmas state
how, in two special cases, the excess of a graph \(G\) can be
bounded in terms of the excesses of its subgraphs.

\begin{lemma}\label{lem:cutvertex}
  Let $G$ be a connected (possibly labelled and/or oriented) graph
  and let $v$ be a cut vertex of $G$. Then
  \(\ex(G)=\smashoperator{\sum\limits_{X\in{}\comp{}(G-\{v\})}}\ex(G[V(X)\cup\{v\}])\).
\end{lemma}
\begin{proof}

\sloppy  

Recall that by definition,
\(\pt(G)=\lambda|E(G)|+\frac{1-\lambda}{2}(|V(G)|-1)\). Observe
first that
\[|E(G)|=\smashoperator{\sum_{X\in\comp(G-\{v\})}}|E(G[V(X)\cup\{v\}])|,\] and
\[{|V(G)|-1\quad=\quad\smashoperator{\sum_{X\in\comp(G-\{v\})}}|V(X)|\quad=\quad\smashoperator{\sum_{X\in\comp(G-\{v\})}}(|V(X)\cup\{v\}|-1)}.\]
Thus
\begin{align*}
  \pt(G)&\quad=\quad\lambda\quad\smashoperator{\sum_{X\in\comp(G-\{v\})}}|E(G[V(X)\cup\{v\}])| 
   \quad+\quad\frac{1-\lambda}{2}\quad\smashoperator{\sum_{X\in\comp(G-\{v\})}}(|V(X)\cup\{v\}|-1)\\
  &\quad=\quad\smashoperator[l]{\sum_{X\in\comp(G-\{v\})}}\left(\lambda|E(G[V(X)\cup\{v\}])|
  +\frac{1-\lambda}{2}(|V(X)\cup\{v\}|-1)\right)\\
  &\quad=\quad\smashoperator{\sum_{X\in\comp(G-\{v\})}}\pt(G[V(X)\cup\{v\}]).
\end{align*}
  
\fussy


We now derive a similar expression for \(\ms(G)\). For each
\(X\in\comp(G-v)\), let $H_X$ be a largest $\LE$-subgraph of
\(G[V(X)\cup\{v\}]\), and let \(H=\union_{X\in \comp(G -
  v)}H_{X}\). Since \(v\) is a cutvertex of graph \(G\), 
we get that every block of \(H\) is a block of some such subgraph
\(H_{X}\). Hence we get---from the block additivity property of
\(\Pi\)---that \(H\) is a \(\Pi\)-subgraph of \(G\). Since no
edge of \(G\) appears in two distinct subgraphs \(H_{X}\), we get
that \(\ms(G)\geq\sum_{X\in\comp(G-\{v\})}\ms(G[V(X)\cup\{v\}])\).

Now consider a largest $\LE$-subgraph $H$ of $G$, and let
\(H_{X}=H[V(X)\cup\{v\}]\) for each \(X\in\comp(G-\{v\})\). Since
\(v\) is a cutvertex of graph \(G\), 
we get that every block of each subgraph \(H_{X}\) is a block of
\(H\). Hence we get---again, from the block additivity property of
\(\Pi\)---that each \(H_{X}\) is a \(\Pi\)-subgraph of the
corresponding subgraph \(G[V(X)\cup\{v\}]\).  Since each edge of
the subgraph \(H\) lies in at least one such \(H_{X}\), we get
that \(\ms(G)\leq\sum_{X\in\comp(G-\{v\})}\ms(G[V(X)\cup\{v\}])\).

Thus
\(\ms(G)=\sum_{X\in\comp(G-\{v\})}\ms(G[V(X)\cup\{v\}])\), and so 
\begin{align*}
\ex(G)&\quad=\quad\ms(G)-\pt(G)\\
&\quad=\quad\smashoperator{\sum_{X\in\comp(G-\{v\})}}\ms(G[V(X)\cup\{v\}])-\smashoperator{\sum_{X\in\comp(G-\{v\})}}\pt(G[V(X)\cup\{v\}])\\
&\quad=\quad\smashoperator{\sum_{X\in\comp(G-\{v\})}}\ms(G[V(X)\cup\{v\}])-\pt(G[V(X)\cup\{v\}])\\
&\quad=\quad\smashoperator{\sum_{X\in\comp(G-\{v\})}}\ex(G[V(X)\cup\{v\}]).\qedhere{}
\end{align*}
\end{proof}

\begin{lemma}\label{lem:half}
  Let $G\in\mathcal{G}$ be a connected graph, and let
  \(V(G)=V_{1}\uplus{}V_{2}\)
  . Let \(c_{1}\) be the number of components of \(G[V_{1}]\) and
  \(c_{2}\) the number of components of \(G[V_{2}]\). If
  \(\ex(G[V_{1}])\geq{}k_{1}\) and $\ex(G[V_2])\geq k_2$, then
  $\ex(G)\geq k_1+k_2-\lp (c_1+c_2-1)$.
\end{lemma}
\begin{proof}
  Let \(E_{i}=E(G[V_{i}])\) for \(i\in\{1,2\}\). Then
  \(E(G)=E_{1}\uplus{}E_{2}\uplus{}E(V{1},V_{2})\). By definition,
  \(\pt(G_{i})=\lambda|E_{i}|+\lp(|V_{i}|-c_{i})\) for
  \(i\in\{1,2\}\), and
  \begin{align*}
    \pt(G)&\quad=\quad\lambda|E(G)|+\lp(|V(G)|-1)\\
    &\quad=\quad\lambda(|E_{1}|+|E_{2}|+|E(V_{1},V_{2})|)+\lp(|V_{1}|+|V_{2}|-1)\\
    &\quad=\quad[\lambda|E_{1}|+\lp(|V_{1}|-c_{1})]+[\lambda|E_{2}|+\lp(|V_{2}|-c_{2})]\\
    &\quad+\quad\lambda|E(V_{1},V_{2})|+\lp(c_{1}+c_{2}-1)\\
    &\quad=\quad\pt(G[V_{1}])+\pt(G[V_{2}])+\lambda|E(V_{1},V_{2})|+\lp(c_{1}+c_{2}-1).
  \end{align*}

  Let $H_i$ be a largest $\LE$-subgraph of $G[V_i]$ for $i \in
  \{1,2\}$. We apply the strong $\lambda$-subgraph extension
  property to the graph
  \((V,E(H_{1})\cup{}E(H_{2})\cup{}E(V_{1},V_{2}))\), its vertex
  partition \((V_{1},V_{2})\), and a weight function which assigns
  unit weights to all its edges. We get that there exists a
  $\LE$-subgraph $H$ of $G$ such that
  \(H=(V,E(H_{1})\uplus{}E(H_{2})\uplus{}F)\), where
  \(F\subseteq{}E(V_{1},V_{2})\) is such that
  \(|F|\geq\lambda|E(V_{1},V_{2})|\). Therefore
  \(\ms(G)\geq\ms(G[V_{1}])+\ms(G[V_{2}])+\lambda|E(V_{1},V_{2})|\). So
  we get that
  \begin{align*}
    \ex(G)&\quad=\quad\ms(G)-\pt(G)\\
    &\quad\geq\quad[\ms(G[V_{1}])+\ms(G[V_{2}])+\lambda|E(V_{1},V_{2})|]\\
    &\quad-\quad[\pt(G[V_{1}])+\pt(G[V_{2}])+\lambda|E(V_{1},V_{2})|+\lp(c_{1}+c_{2}-1)]\\
    &\quad=\quad\ex(G[V_1])+\ex(G[V_2])- \lp(c_1+c_2-1)\\
    &\quad\geq\quad{}k_{1}+k_{2}-\lp(c_1+c_2-1).\qedhere{}
\end{align*}
\end{proof}

We now prove some useful facts about strongly $\lambda$-extendible
properties which diverge on cliques. In particular, we show that
for a property $\Pi$ which diverges on cliques, $\ex(K_j)$
increases as $j$ increases; this motivated our choice of the
name. We also show that $\mak$ is necessarily a constant greater
than $0$.

\begin{lemma}\label{lem:nottoomuchless}
  Let $\ex(K_j)=a\geq\lp$ for some $j\in\mathbb{N}$. Then, for
  every almost-clique $\widetilde{K}$ with at least $j+1$
  vertices, $\ex(\widetilde{K})\geq a-\lp$.
\end{lemma}
\begin{proof}
  Let $G\in\mathcal{G}$ be a graph such that $U(G)=\widetilde{K}$,
  where $\widetilde{K}$ is an almost-clique with at least $j+1$
  vertices. Let \(V'\) be a minimum-sized subset of
  \(V(\widetilde{K})\) such that \(\widetilde{K}-V'\) is a clique.
  Set \(V_{1}\) to be any subset of exactly \(|V(G)|-j\) vertices
  of \(G\) such that (i) \(V'\subseteq{}V_{1}\), and (ii)
  \(G[V_{1}]\) is connected. Set
  \(V_{2}=V(G)\setminus{}V_{1}\). Observe that \(G\) is connected,
  \(V(G)=V_{1}\uplus{}V_{2}\), \(G[V_{1}]\) is connected, and
  \(U(G[V_{2}])=K_{j}\). Further, \(\ex(G[V_{1}])\)
  is---trivially---at least \(0\), and \(\ex(G[V_{2}])\) is---by
  assumption---at least \(a\). So by Lemma~\ref{lem:half}, we get
  that \(\ex(G)\geq{}a-\lp\).
\end{proof}

\begin{lemma}\label{theo:increasingsequence}
  Let $\Pi$ be a strongly $\lambda$-extendible property which
  diverges on cliques, and let $j,a$ be such that
  \(\ex(K_{j})=\lp+a\), $a>0$.  Then \(\ex(K_{rj})\geq\lp+ra\)
  for each \(r\in\mathbb{N}^{+}\).  Furthermore,
  $\lim_{s\rightarrow +\infty}{\ex(K_s)}=+\infty$.
\end{lemma}
\begin{proof}
  We prove the first part of the lemma by induction on \(r\).  The
  claim holds for $r=1$ by assumption. Suppose that the claim
  holds for some \(r\geq1\). We show that it holds for $r+1$ as
  well. Let \(G=K_{(r+1)j}\), and consider a partition of \(V(G)\)
  into two parts \(U,W\) with \(|U|=j,|W|=rj\). Note that
  \(G[U]=K_{j},G[W]=K_{rj}\). By assumption we have that
  \(\ex(G[U])=\lp+a\), and from the induction hypothesis we get
  that \(\ex(G[W])\geq\lp+ra\). Lemma~\ref{lem:half} now tells us
  that \(\ex(G)\geq\lp+(r+1)a\), and this completes the induction
  step.

  Now consider the function \(f:\mathbb{N}^{+}\to\mathbb{R}^{+}\)
  defined as \(f(r)=\ex(K_{rj})\). Our arguments above show also
  that \(f\) is an \emph{unbounded} function. Indeed,
  \(f(r+1)=\ex(K_{(r+1)j})\geq\ex(K_{rj})+\ex(K_{j})-\lp=\ex(K_{rj})+\lp+a-\lp=\ex(K_{rj})+a=f(r)+a\). We
  use this to argue that given any \(x\in\mathbb{R}^{+}\), there
  is an \(r_{x}\in\mathbb{N}^{+}\) such that
  \(\forall{}r\geq{}r_{x}{}\;.\;\ex(K_{r})>x\); this would prove
  the second part of the lemma. So let \(x\in\mathbb{R}^{+}\). We
  choose \(y\in\mathbb{N}^{+}\) such that
  \(f(y)=\ex(K_{yj})=a>x+\lp\). Since \(f\) is unbounded, such a
  choice of \(y\) exists. We set \(r_{x}=yj\), and from
  Lemma~\ref{lem:nottoomuchless} we get that
  \(\forall{}r>r_{x}\;.\;\ex(K_{r})\geq{}a-\lp>x\).
\end{proof}

\begin{lemma}\label{theo:positiveminimum}
  Let $\Pi$ be a strongly $\lambda$-extendible property which
  diverges on cliques. Then $\mak>0$.
\end{lemma}
\begin{proof}
  Since $\Pi$ diverges on cliques, there exist
  \(j\in\mathbb{N}^{+},a\in\mathbb{R}^{+}\) such that
  \(\ex(K_{j})=\lp+a\). Then, by Lemma~\ref{lem:nottoomuchless}, for
  every graph \(G\in{}AK^{+}\) with at least $j+1$ vertices,
  \(\ex(G)\geq{}a\). Now observe that
  \(\{G\in{}AK^{+}:|V(G)|\leq{}j\}\) is a finite set, hence the
  minimum of $\ex(G)$ over this set is defined and is positive. So
  we have that
  \(\mak\geq\min(a,\min_{\{G\in{}AK^{+}:|V(G)|\leq{}j\}}\ex(G))>0\).
\end{proof}


\section{Polynomial kernel for divergence}

In this section we show that \pisub{} has a polynomial kernel, as long as $\Pi$ diverges on cliques and all cliques with at least two vertices have positive excess.

Recall the partition $\mathcal{B}_0, \mathcal{B}_1, \mathcal{B}_2, \mathcal{B}_{\ge 3}$ of the blocks of $G-S$. Since $|S|<\frac{6k}{1-\lambda}$, and the number of cut vertices in $G-S$ is bounded by the number of blocks in $G-S$, it is enough to prove upper bounds on $|\mathcal{B}_0|, |\mathcal{B}_1|, |\mathcal{B}_2|, |\mathcal{B}_
{\ge 3}|$, and $|\inn{B}|$ for every block $B$ in $G-S$.

In order to prove the main result of this section, Theorem \ref{theo:ksquaredkernel}, it is enough to bound $|\mathcal{B}_0|$, together with the number and size of all cliques with positive excess. This is because only the blocks in $\mathcal{B}_0$ may have fewer than two vertices.

We will prove bounds on $|\mathcal{B}_0|, |\mathcal{B}_1|$ and $|\mathcal{B}_{\ge 3}|$ (subject to a reduction rule), and a bound on $|\inn{B}|$ for all blocks $B$ in $G-S$. We do not give a bound on $|\mathcal{B}_2|$ directly, but we do give a bound on the number of cliques with positive excess, which is enough. The bound on $|\mathcal{B}_1|$ is used to bound the number of components in $G-S$, which is required to prove Theorem \ref{theo:ksquaredkernel}. The bounds on $|\mathcal{B}_{\ge 3}|$ and $|\inn{B}|$ are not used by the proof of Theorem \ref{theo:ksquaredkernel}, but they will be useful in the following sections.


\subsection{Bounding $|\mathcal{B}_0|$ and $|\mathcal{B}_1|$}

\begin{definition}
Let $v$ be a cut vertex of $G$ and let $X \subseteq V(G) \setminus \{v\}$ be such that $G[X]$ is a component of $G-\{v\}$ and the underlying graph of $G[X\cup\{v\}]$ is a $2$-connected almost-clique. Then we say that $G[X\cup\{v\}]$ is a {\em dangling component} with root $v$.
\end{definition}

 To bound the number of isolated and leaf blocks in $G-S$, we require the following reduction rule.
 
\begin{krule}\label{krule:zeroexcess}
 Let $G\in\mathcal{G}$ be a connected graph with at least two $2$-connected components and let $G'$ be a 
dangling component.
 Then if $\ex(G')=0$, delete $G'-\{v\}$ (where $v$ is root of $G'$) and leave $k$ the same.
\end{krule}

\begin{lemma}
 Reduction Rule \ref{krule:zeroexcess} is valid.
\end{lemma}
\begin{proof}
 Let $G''$ be the graph obtained after an application of the rule. By Lemma \ref{lem:cutvertex}, $\ms(G)=\ms(G')+\ms(G'')$ and $\pt(G)=\pt(G')+\pt(G'')$, which is enough to ensure that $\ex(G)=\ex(G'')$. 
\end{proof}

\begin{lemma}
 Reduction Rule \ref{krule:zeroexcess} applies in polynomial time if $\Pi$ diverges on cliques.
\end{lemma}
\begin{proof}
  In polynomial time it is possible to find all $2$-connected
  components of $G$ and to check which ones have an underlying
  graph which is an almost-clique.  Thus, in polynomial time we
  can find all dangling components.  Now, we claim that in
  constant time it is possible to evaluate whether their excess is
  zero. In fact, by the definition of divergence on cliques, it
  holds that $\ex(K_j)>\lp$ for some $j$. Given a subgraph $G'$
  whose underlying graph is an almost-clique, if $G'$ has at least
  $j+1$ vertices Lemma \ref{lem:nottoomuchless} ensures that
  $\ex(G')>0$. On the other hand, if $G'$ has at most $j$
  vertices, a brute force algorithm which finds the largest
  $\Pi$-subgraph of $G'$ runs in time $O(2^{j^2})$, where $j$ is a
  constant which only depends on $\Pi$.  
\end{proof}

\begin{lemma}\label{theo:danglingbound}
 Let $\Pi$ be a strongly $\lambda$-extendible property which diverges on cliques and let $G$ be a connected graph reduced under Reduction Rule \ref{krule:zeroexcess}. Then the number of 
dangling components
 is less than $\frac{k}{\mak}$, or the instance is a YES-instance.
\end{lemma}
\begin{proof}
 Let $B_1,\dots,B_l$ be the dangling components of $G$. Since the graph is reduced under Reduction Rule \ref{krule:zeroexcess}, $\ex(B_i)>0$ for every $1\leq i\leq l$. Since $\Pi$ diverges on cliques, Lemma \ref{theo:positiveminimum} ensures that $\mak>0$.
Let $G'=G-(\cup_{i=1}^l(\inn{(B_i)})$.

 By Lemma \ref{lem:cutvertex}, $\ex(G)=\ex(G')+\sum_{i=1}^l\ex(B_i)\geq 0+\mak l$. Then if $l\geq\frac{k}{\mak}$ the instance is a YES-instance.  
\end{proof}


\begin{theorem}\label{theo:leafbound}
 Let $\Pi$ be a strongly $\lambda$-extendible property which diverges on cliques and let $G$ be a connected graph reduced under Reduction Rule \ref{krule:zeroexcess}. If there exists $s\in S$ such that $\sum_{B\in\mathcal{B}}|N_G(\inn{B})\cap\{s\}|$ is at least $(\frac{16}{1-\lambda}+\frac{2}{\mak})k-2$, then the instance is a YES-instance.
\end{theorem}
\begin{proof}
 Let $U=\{s\}$. For every block $B$ of $G-S$ such that $|N_G(\inn{B})\cap\{s\}|=1$, pick a vertex in $N(s)\cap\inn{B}$ and add it to $U$. Since the vertices are chosen in the interior of different blocks of $G-S$, $G[U]$ is a star and therefore it is in $\Pi$ by block additivity. By Lemma \ref{lem:cutvertex}, $\ex(G[U])=\lp d$, where $d$ is the degree of $s$ in $G[U]$. Let $c$ be the number of components of $G-U$, and assume that $U$ is constructed such that $d$ is maximal and $c$ is minimal. By Lemma \ref{lem:half}, $\ex(G)\geq\lp(d-c)$.

 We will now show that $c$ is bounded. The number of components of $G-U$ which contain a vertex of $S\setminus\{s\}$ is bounded by $(|S| - 1)<\frac{6k}{1-\lambda}-1$. In addition, the number of components of $G-S$ which contain at least two blocks from which a vertex has been added to $U$ is at most $\frac{d}{2}$.

 Now, let $T$ be a component of $G-S$ such that in the graph $G$, no vertex in $T-U$ has a neighbor in $S\setminus\{s\}$ and $|U\cap V(T)|=1$. Firstly, suppose that $T$ contains only one block $B$ of $G-S$. Let $\{v\}=U\cap V(T)$. Note that, by the current assumptions, $N(S\setminus\{s\}) \cap V(T) \subseteq\{v\}$. If $v$ is the only neighbor of $s$ in $T$, then it is a cut vertex in $G$, hence $B$ is a dangling component of $G$. If $s$ has another neighbor $v'$ in $T$ and $v$ has no neighbor in $S$ different from $s$, then $s$ is a cut vertex, therefore $G[V(B)\cup\{s\}]$ is a dangling component. Finally, if $v$ has at least two neighbors in $S$ and $s$ has at least another neighbor $v'$ in $T$, let $U'$ be the star obtained by replacing $v$ with $v'$ in $U$, and observe that $T$ is connected to $S\setminus\{s\}$ in $G-U'$, contradicting the minimality of $c$.


 Now, suppose that $T$ contains at least two blocks of $G-S$. In this case, every block except $B$ does not contain neighbors of $S$. In particular, this holds for at least one leaf block $B'$ in $T$. Hence, $B'$ is a dangling component.

 This ensures that carefully choosing the vertices of $U$ we may assume that any component of $G-U$ 
still contains a vertex of $S\setminus\{s\}$, or contains at least two blocks from which a vertex of $U$ has been chosen, or contains part of a dangling component. Hence, the number of components of $G-U$ is at most $\frac{6k}{1-\lambda}-1+\frac{d}{2}+\frac{k}{\mak}$.

 Therefore, if $d\geq(\frac{16}{1-\lambda}+\frac{2}{\mak})k-2$, then $\ex(G)\geq k$.  
\end{proof}

\begin{corollary}\label{cor:sneighborbound}
 Let $\Pi$ be a strongly $\lambda$-extendible property which diverges on cliques and let $G$ be a connected graph reduced under Reduction Rule \ref{krule:zeroexcess}. If $\sum_{s\in S}\sum_{B\in\mathcal{B}}|N_G(\inn{B})\cap\{s\}|$ is at least $((\frac{16}{1-\lambda}+\frac{2}{\mak})k-2)\frac{6k}{1-\lambda}$, then the instance is a YES-instance.
\end{corollary}

\begin{corollary}\label{cor:danglingleafbound}
 Let $\Pi$ be a strongly $\lambda$-extendible property which diverges on cliques and let $G$ be a connected graph reduced under Reduction Rule \ref{krule:zeroexcess}. Then $|\mathcal{B}_0|+|\mathcal{B}_1|\leq((\frac{16}{1-\lambda}+\frac{2}{\mak})k-2)\frac{6k}{1-\lambda}+\frac{k}{\mak}$, or the instance is a {\sc Yes}-instance.
\end{corollary}
\begin{proof}
 Note that every isolated or leaf block either has a neighbor of $S$ in its interior or is a dangling block. The claim follows from Lemma \ref{theo:danglingbound} and Corollary \ref{cor:sneighborbound}. 
\end{proof}

\begin{corollary}\label{cor:gminuscomponents}
  Let $\Pi$ be a strongly $\lambda$-extendible property which diverges on cliques and let $G$ be a connected graph reduced under Reduction Rule \ref{krule:zeroexcess}. Then either $G-S$ has at most $((\frac{16}{1-\lambda}+\frac{2}{\mak})k-2)\frac{6k}{1-\lambda} + \frac{k}{\mak}$ components,  or the instance is a {\sc Yes}-instance. 
\end{corollary}
\begin{proof}
 Since a component of $G-S$ contains at least one block from $\mathcal{B}_0\cup\mathcal{B}_1$, the result follows applying Corollary \ref{cor:danglingleafbound}.
\end{proof}


\subsection{Bounding blocks with positive excess}

\begin{lemma}\label{lem:positiveexcessblock}
  Let $\Pi$ be a strongly $\lambda$-extendible property which diverges on cliques. If $G-S$ contains at least $((\frac{16}{1-\lambda}+\frac{2}{\mak})k-1)\frac{6k}{\mak(1-\lambda)} + \frac{k}{(\mak)^2} + \frac{k-1}{\mak}$ blocks with positive excess, then the instance is a {\sc Yes}-instance.
\end{lemma}
\begin{proof}
 Let $d$ be the number of blocks in $G-S$ with positive excess, and let $G'$ be the union of all components of $G-S$ which contain a block with positive excess. Observe that by Corollary \ref{cor:gminuscomponents}, we may assume $G'$ has at most $((\frac{16}{1-\lambda}+\frac{2}{\mak})k-2)\frac{6k}{1-\lambda} + \frac{k}{\mak}$ components. Observe that by repeated use of Lemma \ref{lem:cutvertex}, $\ex(G') \ge d\mak$. Now let $G''=G - G'$, and observe that $G''$ has at most $|S|\le \frac{6k}{1-\lambda}$ components.
Then by Lemma \ref{lem:half}, $\ex(G) \ge d\mak - \lp (((\frac{16}{1-\lambda}+\frac{2}{\mak})k-2)\frac{6k}{1-\lambda} + \frac{k}{\mak}+\frac{6k}{1-\lambda}-1)$.
Therefore if $d \ge ((\frac{16}{1-\lambda}+\frac{2}{\mak})k-1)\frac{6k}{\mak(1-\lambda)} + \frac{k}{(\mak)^2} + \frac{k-1}{\mak}$, the instance is a {\sc Yes}-instance. 
\end{proof}

%

\subsection{Bounding $|\mathcal{B}_{\geq 3}|$}

The following lemma is not required to prove the last theorem in this section, but it will be used in Section 6.

\begin{lemma}\label{lem:nonleafblocks}
 Let $G\in\mathcal{G}$ be a connected graph and $S\subseteq V(G)$ be such that $G-S$ is a forest of cliques. Then $|\mathcal{B}_{\geq 3}|\leq 3|\mathcal{B}_1|$.
\end{lemma}
\begin{proof}
 The proof is by induction on $|\mathcal{B}|$. We may assume that $G-S$ is connected, otherwise we can prove the bound separately for every component. If $|\mathcal{B}|=1$, then $|\mathcal{B}_{\geq 3}|=0$ and the bound trivially holds. Suppose now that $|\mathcal{B}|=l+1\geq 2$ and that the bound holds for every tree of cliques with at most $l$ blocks. Let $B\in\mathcal{B}$ be a leaf block and let $v$ be its root. Let $G'=G-(V(B)\setminus\{v\})$. $G'-S$ is a tree of cliques with at most $l$ blocks, so by induction hypothesis $|\mathcal{B}_{\geq 3}'|\leq 3|\mathcal{B}_1'|$. Now, note that only block neighbors of $B$ can be influenced by the removal of $B$: in other words, if a block $B'$ is not a block neighbor of $B$ and $B'\in\mathcal{B}_i$, then $B'\in\mathcal{B}_i'$ for every $i\in\{1,2,\geq 3\}$.

 We distinguish three cases. Recall that $Q$ is the set of
cutvertices of $G - S$. Let $Q'$ be the set of cutvertices of $G' - S$.

\noindent{\bf Case 1 }({\small \em $B$ has at least three block neighbors){\bf :}} In this case it holds that $Q=Q'$, which ensures that the removal of $B$ does not increase the number of leaf blocks, that is, $|\mathcal{B}_1'|=|\mathcal{B}_1|-1$. In addition, if a block neighbor $B'$ of $B$ is in $\mathcal{B}_{\geq 3}$, then it is in $\mathcal{B}_{\geq 3}'$, which means that $|\mathcal{B}_{\geq 3}'|=|\mathcal{B}_{\geq 3}|$. Therefore in this case, using the induction hypothesis it follows that $|\mathcal{B}_{\geq 3}|=|\mathcal{B}_{\geq 3}'|\leq 3|\mathcal{B}_1'|\leq 3|\mathcal{B}_1|-3$.

\medskip

\noindent{\bf Case 2 }({\small \em $B$ has two block neighbors){\bf :}} As in the previous case $Q=Q'$, hence $|\mathcal{B}_1'|=|\mathcal{B}_1|-1$. On the other hand, if a block neighbor $B'$ of $B$ is in $\mathcal{B}_{\geq 3}$, it could be the case that $B'$ is in $\mathcal{B}_2'$. Therefore, $|\mathcal{B}_{\geq 3}'|\geq|\mathcal{B}_{\geq 3}|-2$. Using the induction hypothesis it follows that $|\mathcal{B}_{\geq 3}|\leq|\mathcal{B}_{\geq 3}'|+2\leq 3|\mathcal{B}_1'|+2\leq 3|\mathcal{B}_1|$.

\medskip

\noindent{\bf Case 3 }({\small \em $B$ has exactly one block neighbor){\bf :}} Let $B'$ be the only block neighbor of $B$. Again, we distinguish three cases. If $B'\in\mathcal{B}_1$, then $B$ and $B'$ are the only blocks of $G-S$ and $|\mathcal{B}_{\geq 3}|=0$, therefore the bound trivially holds. If $B'\in\mathcal{B}_2$, then $B'$ is a leaf block in $G'-S$, hence $|\mathcal{B}_1|=|\mathcal{B}_1'|$ and $|\mathcal{B}_{\geq 3}|=|\mathcal{B}_{\geq 3}'|$: the bound follows using the induction hypothesis.

Lastly, let $B'\in\mathcal{B}_{\geq 3}$. If $|V(B')\cap Q|\geq 3$, then $|\mathcal{B}_{\geq 3}'|\geq|\mathcal{B}_{\geq 3}|-1$ and $|\mathcal{B}_1'|=|\mathcal{B}_1|-1$: therefore, by induction hypothesis, $|\mathcal{B}_{\geq 3}|\leq |\mathcal{B}_{\geq 3}'|+1\leq 3|\mathcal{B}_1'|+1\leq 3|\mathcal{B}_1|$. Otherwise, $|V(B')\cap Q|=2$ and $B'$ is a leaf block in $G'-S$. In this case, $|\mathcal{B}_1'|=|\mathcal{B}_1|$ and $|\mathcal{B}_{\geq 3}'|=|\mathcal{B}_{\geq 3}|-1$. Now, consider the graph $G''=G'-(V(B')\setminus\{v'\})$, where $v'$ is the root of $B'$ in $G'-S$. Removing $B'$ from $G'$ corresponds either to case 1 or 2, hence it holds that $|\mathcal{B}_1''|=|\mathcal{B}_1'|-1$ and $|\mathcal{B}_{\geq 3}''|\geq |\mathcal{B}_{\geq 3}'|-2$. Therefore, using the induction hypothesis on $G''-S$ (which is a tree of cliques with $l-1$ blocks) it follows that $|\mathcal{B}_{\geq 3}|=|\mathcal{B}_{\geq 3}'|+1\leq|\mathcal{B}_{\geq 3}''|+3\leq 3|\mathcal{B}_1''|+3=3|\mathcal{B}_1'|=3|\mathcal{B}_1|$, which concludes the proof.  
\end{proof}

\begin{corollary}\label{cor:nonpathblocksbound}
 Let $\Pi$ be a strongly $\lambda$-extendible property which diverges on
cliques and let $G$ be a connected graph reduced under Reduction Rule
\ref{krule:zeroexcess}. Then
$|\mathcal{B}_0|+|\mathcal{B}_1|+|\mathcal{B}_{\geq 3}|\leq
4(((\frac{16}{1-\lambda}+\frac{2}{\mak})k-2)\frac{6k}{1-\lambda}+\frac{k}{\mak})$,
or the instance is a YES-instance.
\end{corollary}
\begin{proof}
 The bound follows from Corollary \ref{cor:danglingleafbound} and Lemma
\ref{lem:nonleafblocks}.
\end{proof}

\subsection{Bounding $|\inn{B}|$}
\begin{lemma}\label{lem:int}
 Let $\Pi$ be a strongly $\lambda$-extendible property which diverges on cliques,
and let $j,a$ be such that $\ex(K_j)=\frac{1-\lambda}{2}+a$, $a>0$.
If $|\inn{B}|\geq \lceil \frac{4k}{a}+ \frac{1-\lambda}{2a}\rceil j$ for any block $B$ of $G-S$, then the instance is a {\sc Yes}-instance. 
\end{lemma}
\begin{proof}
 Note that $G-\inn{B}$ has at most $\frac{6k}{1-\lambda}+1$ components, since every component of $G-S$ which does not contain $B$ still has an edge to a vertex of $S$, while the only component of $G-S$ that could be disconnected from $S$ is the one containing $B$. Therefore, if $\ex(\inn{B})\geq\lp(\frac{6k}{1-\lambda}+1) + k = 4k+\lp$, Lemma \ref{lem:half} ensures that we have a {\sc Yes}-instance. 

By Lemma \ref{theo:increasingsequence}, if $r$ is an integer such that $r \ge \frac{4k}{a} + \frac{1-\lambda}{2a}$, then $\ex(K_{rj}) \geq 4k + (1-\lambda)$. This shows that if $|\inn{B}|= rj$ then $\ex(\inn{B}) \geq 4k + (1-\lambda)$.
Furthermore, by Lemma \ref{lem:nottoomuchless}, if $|\inn{B}|\ge rj$ then $\ex(\inn{B})\geq 4k + \lp$. 
Thus, if  $|\inn{B}|\geq\lceil \frac{4k}{a}+ \frac{1-\lambda}{2a}\rceil j$  we have a {\sc Yes}-instance. 
\end{proof}

\subsection{Combined kernel}

\begin{theorem}\label{theo:ksquaredkernel}
 Let $\Pi$ be a strongly $\lambda$-extendible property which diverges on
cliques, and suppose $\ex(K_i)>0$ for all $i\ge 2$.
Then \pisub{} has a kernel with at most \(\cO(k^{2})\) vertices.
\end{theorem}
\begin{proof}
 Let $j\in\mathbb{N}$ be such that $\ex(K_j)=\lp+a$, where $a>0$. Let
$V(G-S)=V'\cup V''\cup V''$, where $V'$ is the set of vertices contained in blocks with exactly one vertex,
$V''$ is the set of vertices contained in
blocks with between $2$ and $j-1$ vertices and $V'''$ is the set of vertices
contained in blocks with at least $j$ vertices (note that in general
$V''\cap V'''\neq\emptyset$). 
Observe that the blocks containing $V'$ are isolated blocks, therefore by Corollary
\ref{cor:danglingleafbound} there exists a constant $c_1$ depending only on $\Pi$ such that $|V'|\leq c_1k^2$,
 or the instance is a YES-instance.
By Lemma \ref{lem:positiveexcessblock}, there
exists a constant $c_2$ depending only on $\Pi$ such that $|V''|\leq
c_2(j-1)k^2$, or the instance is a YES-instance.

 Now, let $\mathcal{B}''$ be the set of blocks of $G-S$ which contain at
least $j$ vertices. For every block $B\in\mathcal{B}''$, let $rj+l$ be
the number of its vertices, where $0\leq l<j$. Note that, by Lemma
\ref{theo:increasingsequence} and Lemma \ref{lem:half}, $\ex(B)\geq ra$.
Let $d=\sum_{B\in\mathcal{B}''}r$ and let $G''$ be the union of all
components of $G-S$ which contain a block in $\mathcal{B}''$. By
Corollary \ref{cor:gminuscomponents}, we may assume that $G''$ has at
most $((\frac{16}{1-\lambda}+\frac{2}{\mak})k-2)\frac{6k}{1-\lambda} +
\frac{k}{\mak}$ components. Furthermore, by repeated use of Lemma
\ref{lem:cutvertex}, we get that $\ex(G'')\geq da$. Observe that $G-G''$
has at most $|S|\leq \frac{6k}{1-\lambda}$ components: then, by Lemma
\ref{lem:half}, $\ex(G)\geq da-\lp
(((\frac{16}{1-\lambda}+\frac{2}{\mak})k-2)\frac{6k}{1-\lambda} +
\frac{k}{\mak}+\frac{6k}{1-\lambda}-1)$. Therefore if $d \ge
((\frac{16}{1-\lambda}+\frac{2}{\mak})k-1)\frac{6k}{a(1-\lambda)} +
\frac{k}{(a\mak)} + \frac{k-1}{a}$, the instance is a {\sc Yes}-instance.
Otherwise, $|V'''|\leq 2dj\leq c_3jk^2$, where $c_3$ is a constant which
depends only on $\Pi$, which means that $|V(G)|\leq
\frac{6k}{1-\lambda}+(c_1+c_2(j-1)+c_3j)k^2$.
\end{proof}

\newcommand{\tor}{\stackrel{\rightarrow}{K}_3}
\newcommand{\tuo}{\stackrel{\nrightarrow}{K}_3}

\section{Kernel when $\lambda\neq\frac{1}{2}$ or $K_3\in\Pi$}



\begin{lemma}\label{lem:exK3}
 Let $\Pi$ be a strongly $\lambda$-extendible property with $\lambda\neq\frac{1}{2}$. Then $\ex(K_3)\geq 1-2\lambda$ and, if $\lambda>\frac{1}{2}$, $\ex(K_3)=2-2\lambda$. In particular, $\ex(K_3)>0$ in every case.
\end{lemma}
\begin{proof}
 Note that $\beta(G)\geq 2$ for any connected graph $G\in\mathcal{G}$ with at least two edges, because any graph whose underlying graph is a path on three vertices is in $\Pi$ by inclusiveness and block additivity. Therefore, $\beta(K_3)\geq 2$, which ensures that $\ex(K_3)\geq 2-(3\lambda+\frac{1-\lambda}{2}2)=1-2\lambda$, which is strictly greater than zero if $\lambda<\frac{1}{2}$.

 Now, assume $\lambda>\frac{1}{2}$, let $G\in\mathcal{G}$ be such that $U(G)=K_3$ and let $V(G)=\{v_1,v_2,v_3\}$. Consider $U=\{v_1,v_2\}$ and $W=\{v_3\}$ and note that $G[U],G[W]\in\Pi$ by inclusiveness. Then, by the strong $\lambda$-subgraph extension property, it holds that $G\in\Pi$, which ensures that $\beta(K_3)=3$. This means that $\ex(K_3)=3-(3\lambda+\frac{1-\lambda}{2}2)=2-2\lambda>0$. 
\end{proof}

\begin{lemma}\label{theo:easylambda}
 Let $\Pi$ be a strongly $\lambda$-extendible property. If $\lambda\neq\frac{1}{2}$, then $\ex(K_3)>\frac{1-\lambda}{2}$ or $\ex(K_4)>\frac{1-\lambda}{2}$. In particular, $\Pi$ diverges on cliques.
\end{lemma}
\begin{proof}
 If $\lambda>\frac{1}{2}$, then by Lemma \ref{lem:exK3} $\ex(K_3)=2-2\lambda>\lp$. If $\lambda<\frac{1}{3}$, then by Lemma \ref{lem:exK3} $\ex(K_3)\geq 1-2\lambda$, which is greater than $\frac{1-\lambda}{2}$. Lastly, if $\frac{1}{3}\leq\lambda<\frac{1}{2}$, let $G\in\mathcal{G}$ be such that $U(G)=K_4$ and let $V(G)=\{v_1,v_2,v_3,v_4\}$. Consider $U=\{v_1,v_2\}$ and $W=\{v_3,v_4\}$ and note that $G[U],G[W]\in\Pi$ by inclusiveness. By the strong $\lambda$-subgraph extension property, it holds that $\beta(G)\geq 4$, since $\lambda>\frac{1}{4}$. Therefore, $\ex(K_4)\geq 4-6\lambda-\frac{1-\lambda}{2}3=\frac{5}{2}-\frac{9}{2}\lambda$ which is greater than $\frac{1-\lambda}{2}$. 
\end{proof}

\begin{lemma}\label{lem:triangle}
 Let $\Pi$ be a strongly $\lambda$-extendible property. If $K_3\in\Pi$, then $\ex(K_3)>\frac{1-\lambda}{2}$. In particular, $\Pi$ diverges on cliques.
\end{lemma}
\begin{proof}
 If $K_3\in\Pi$, then $\beta(K_3)=3$, which means that $\ex(K_3)=2-2\lambda>\lp$. 
\end{proof}

\begin{lemma}\label{lem:positivecliques}
  Let $\Pi$ be a strongly $\lambda$-extendible property. If $\lambda\neq\frac{1}{2}$ or $K_3\in\Pi$, then $\ex(K_i)>0$ for all $i\ge 2$.
\end{lemma}
\begin{proof}
 By Lemma \ref{theo:easylambda} and Lemma \ref{lem:triangle}, $\ex(K_3)>\lp$ or $\ex(K_4)>\lp$. In the first case, by Lemma \ref{lem:nottoomuchless}, it holds that $\ex(K_j)>0$ for all $j\geq 4$, while in the second case, using the same Lemma, $\ex(K_j)>0$ for all $j\geq 5$. In addition, by Lemma \ref{lem:exK3}, $\ex(K_3)>0$. Finally, $\ex(K_2)=1-(\lambda+\lp)=\lp>0$. 
\end{proof}

\begin{theorem}\label{thm:easycase}
  Let $\Pi$ be a strongly $\lambda$-extendible property. If $\lambda\neq\frac{1}{2}$ or $K_3\in\Pi$, then \APT{} has a kernel with \(\cO(k^{2})\) vertices.
\end{theorem}
\begin{proof}
 By Lemma \ref{theo:easylambda} or Lemma \ref{lem:triangle}, $\Pi$ diverges on cliques. Furthermore, by Lemma \ref{lem:positivecliques}, $\ex(K_i)>0$ for all $i\ge 2$. Then, by Theorem \ref{theo:ksquaredkernel}, \APT{} has a kernel with at most \(\cO(k^{2})\) vertices. 
\end{proof}

By Theorem \ref{thm:easycase}, the only remaining cases to consider are those for which $\lambda = 1/2$ and $\Pi$ does not contain all triangles. We do this in the following section.

\section{Kernel when $\lambda=\frac{1}{2}$}

\begin{definition}\label{def:hereditary}
A graph property $\LE$ is \emph{hereditary} if, for any graph $G$ and vertex-induced subgraph $G'$ of $G$, if $G \in \LE$ then $G' \in \LE$.
\end{definition}

\begin{theorem}\label{theo:MaxCutEverywhere}
 Let $\Pi$ be a strongly $\lambda$-extendible property with $\lambda=\frac{1}{2}$. Suppose $\Pi$ is hereditary and $G\notin\Pi$ for any $G\in\mathcal{G}$ such that $U(G)=K_3$. Then $\Pi=\{G\in\mathcal{G}:G\ is\ bipartite\}$.
\end{theorem}
\begin{proof}
 First, assume for the sake of contradiction that $\Pi$ contains a non-bipartite graph $H$. Then $H$ contains an odd cycle $C_l$. By choosing $l$ as small as possible we may assume that $C_l$ is a vertex-induced subgraph of $H$. Then, since $\LE$ is hereditary, $C_l$ is in $\Pi$. Note that if $l=3$, then $U(C_3)=K_3$, so this is not the case. Consider the graph $H'$ obtained from $C_l$ adding a new vertex $v$ and an edge from $v$ to every vertex of $C_l$. Since both $C_l$ and $K_1=\{v\}$ are in $\Pi$, by the strong $\lambda$-subgraph extension property we can find a subgraph of $H'$ which contains $C_l$, $v$ and at least half of the edges between $v$ and $C_l$. Since $l$ is odd, for any choice of $\frac{l+1}{2}$ edges there are two of them, say $e_1=vx$ and $e_2=vy$, such that the edge $xy$ is in $C_l$. Therefore, since $\Pi$ is hereditary, $H'[v,x,y]\in\Pi$, which leads to a contradiction, as $U(H'[v,x,y])=K_3$.

 Now, we will show that all connected bipartite graphs are in $\Pi$, independently from any possible labelling and/or orientation. We will proceed by induction. The claim is trivially true for $j=1,2$. Assume $j\geq 3$ and that every bipartite graph with $l<j$ vertices is in $\Pi$. Consider any connected bipartite graph $H$ with $j$ vertices. Consider a vertex $v$ such that $H'=H-\{v\}$ is connected. By induction hypothesis, $H'\in\Pi$. Consider the graph $H''$ obtained from $H'$ and $G_2$, where $G_2$ is any graph in $\mathcal{G}$ with $U(G_2)=K_2$ (let $V(G_2)=\{v_1,v_2\}$), adding an edge from $v_i$ to $w\in V(H')$ if and only if in $H$ there is an edge from $v$ to $w$. Colour red the edges from $v_1$ and blue the edges from $v_2$.

 Since $G_2\in\Pi$ by inclusiveness and $H'\in\Pi$, by the strong $\lambda$-subgraph extension property there must exist a subgraph $\widetilde{H}$ of $H''$ which contains $G_2$, $H'$ and at least half of the edges between them. Note that the red edges are exactly half of the edges and that if $\widetilde{H}$ contains all of them and no blue edges, then we can conclude that $H$ is in $\Pi$ by block additivity. The same holds if $\widetilde{H}$ contains every blue edge and no red edge.

If, on the contrary, $\widetilde{H}$ contains one red and one blue edge, we will show that it contains a vertex-induced cycle of odd length, which leads to a contradiction according to the first part of the proof. First, suppose that both these edges contain $w\in V(H')$: if this happens, $\widetilde{H}$ contains a cycle of length $3$ as a vertex-induced subgraph.

Now, suppose $\widetilde{H}$ contains a red edge $v_1w_1$ and a blue edge $v_2w_2$. Note that $w_1$ and $w_2$ are in the same partition and, since $H'$ is connected, there is a path from $w_1$ to $w_2$ which has even length. Together with $v_1w_1$, $v_2w_2$ and $v_1v_2$, this gives a cycle of odd length. Choosing the shortest path between $w_1$ and $w_2$, we may assume that the cycle is vertex-induced.

Thus, we conclude that the only possible choices for $\widetilde{H}$ are either picking the red edges or picking the blue edges, which concludes the proof. 
\end{proof}

The above theorem is of interest due to the following theorem:

\begin{theorem}\cite{CrowstonGJM12}\label{thm:maxCutKernel}
\mcatlb{} has a kernel with \(\cO(k^{3})\) vertices.
\end{theorem}

\subsection{Simple graphs}

In this part, we assume that $\mathcal{G}$ is the class of simple graphs, that is, without any labelling or orientation. Note that, in this case, there is only one graph, up to isomorphism, whose underlying graph is $K_3$ (namely, $K_3$ itself).

\begin{theorem}\label{thm:simple}
  Let $\Pi$ be a strongly $\lambda$-extendible property on simple graphs, with $\lambda=\frac{1}{2}$, and suppose $\Pi$ is hereditary. Then \APT{} has a kernel with \(\cO(k^{2})\) or \(\cO(k^{3})\) vertices.
\end{theorem}
\begin{proof}
 If $K_3\notin\Pi$, by Theorem \ref{theo:MaxCutEverywhere} $\Pi$ is equal to Max Cut and therefore by Theorem \ref{thm:maxCutKernel} it admits a kernel with \(\cO(k^{3})\). On the other hand, if $K_3\in\Pi$, then by Theorem \ref{thm:easycase} $\Pi$ admits a kernel with \(\cO(k^{2})\) vertices. 
\end{proof}

\subsection{Oriented graphs}

In this part, we assume that $\mathcal{G}$ is the class of oriented graphs, without any labelling.

\begin{definition}
 Let $\tor\in\mathcal{G}$ be such that $U(\tor)=K_3$, $V(\tor)=\{v_1,v_2,v_3\}$ and $o((v_i,v_{i+1}))=\ >$ for $1\leq i\leq 2$ and $o((v_1,v_3))=\ <$. We will call $\tor$ the {\em oriented} triangle.

 Similarly, let $\tuo\in\mathcal{G}$ be such that $U(\tuo)=K_3$, $V(\tuo)=\{u_1,u_2,u_3\}$ and $o((u_i,u_j))=\ >$ for every $i<j$, $1\leq i,j\leq 3$. We will call $\tuo$ the {\em non-oriented} triangle.
\end{definition}

It is not difficult to see that, up to isomorphism, $\tor$ and $\tuo$ are the only graphs in $\mathcal{G}$ with $K_3$ as underlying graph.

\begin{lemma}\label{lem:torthenall}
 Let $\Pi$ be a strongly $\lambda$-extendible property on oriented graphs, with $\lambda=\frac{1}{2}$, and suppose $\Pi$ is hereditary. If $\tor\in\Pi$, then $\tuo\in\Pi$.
\end{lemma}
\begin{proof}
 Consider the graph $H$ obtained by adding a vertex $v$ to $\tor$ and an edge from $v$ to every vertex of $\tor$, such that $o((v,v_i))=\ >$ for every $1\leq i\leq 3$. Since $\tor\in\Pi$, by the strong $\lambda$-subgraph extension property there exists a $\Pi$-subgraph $H'$ of $H$ which contains $\tor$, $v$ and at least two edges between $\tor$ and $v$: without loss of generality, assume these edges are $vv_1$ and $vv_2$. Then since $\LE$ is hereditary $H'[v,v_1,v_2]\in\Pi$ and note that $H'[v,v_1,v_2]$ is isomorphic to $\tuo$. 
\end{proof}

\begin{lemma}\label{lem:tuothendiverges}
 Let $\Pi$ be a strongly $\lambda$-extendible property on oriented graphs, with $\lambda=\frac{1}{2}$, and suppose $\Pi$ is hereditary. If $\tuo\in\Pi$, then $ex(K_4)>\frac{1}{4}$. In particular, $\Pi$ diverges on cliques.
\end{lemma}
\begin{proof}
 Let $H\in\mathcal{G}$ be such that $U(H)=K_4$ and let $V(H)=\{w_1,w_2,w_3,w_4\}$. If $H[w_1,w_2,w_3]=\tor$ and $H[w_2,w_3,w_4]=\tor$, then $H[w_1,w_2,w_4]=\tuo$. Hence, for any orientation on the edges of $H$, the graph contains $\tuo$ as a vertex-induced subgraph. Now, since $\tuo\in\Pi$, by the strong $\lambda$-subgraph extension property there exists a $\Pi$-subgraph of $H$ which contains at least $5$ edges, which means that $\beta(H)\geq 5$. This ensures that $\ex(K_4)\geq 5-(3+\frac{3}{4})=\frac{5}{4}$, which concludes the proof. 
\end{proof}

\begin{lemma}\label{lem:tuothenorientedtight}
 Let $\Pi$ be a strongly $\lambda$-extendible property on oriented graphs, with $\lambda=\frac{1}{2}$, and suppose $\Pi$ is hereditary, $\tor\notin\Pi$ and $\tuo\in\Pi$. Then $\ex(K_j)>0$ for every $j\neq 3$.
\end{lemma}
\begin{proof}
 Note that $\ex(K_4)>\frac{1}{4}$, then by Lemma \ref{lem:nottoomuchless} $\ex(K_j)>0$ for every $j\geq 4$. In addition, $\ex(K_2)=\frac{1}{4}$. 
\end{proof}

Let $\mathcal{B}_2^0$ be the subset of $\mathcal{B}_2$ which contains all the blocks with excess zero, and have no internal vertices in $N(S)$.
Let $Q_0$ denote the set of cut vertices of $G-S$ which only appear in blocks in $\mathcal{B}_2^0$.
Note that every vertex in $Q_0$ appears in exactly two blocks in $\mathcal{B}_2^0$.

\begin{lemma}\label{theo:neighborsbound}
 Let $\Pi$ be a strongly $\lambda$-extendible property on oriented graphs, with $\lambda=\frac{1}{2}$, and suppose $\Pi$ is hereditary, $\tor\notin\Pi$ and $\tuo\in\Pi$. Let $(G,k)$ be an instance of \APT{} reduced by Reduction Rule \ref{krule:zeroexcess}.
For any $s\in S$, if $|Q_0 \cap N(s)| \ge ((32+\frac{2}{\mak})k-2)48k-\frac{4k}{\mak} + 4k$, then the instance is a {\sc Yes}-instance.

%
%
\end{lemma}
\begin{proof}
 First, note that all the blocks in $\mathcal{B}_2^0$ are isomorphic to $\tor$ by Lemma \ref{lem:tuothenorientedtight}.
Observe that every vertex in $Q_0$ has at most two neighbours in $Q_0$. Since all vertices in $Q_0$ are cut vertices of $G-S$, it follows that $G[Q_0\cap N(s)]$ is a disjoint union of paths. It follows that we can find a set $Q_0' \subseteq Q_0\cap N(s)$ such that $|Q_0'|\ge \frac{|Q_0 \cap N(s)|}{2}$ and $Q_0'$ is an independent set.

For each $v \in Q_0'$, let $B_1, B_2$ be the two blocks in $\mathcal{B}_2^0$ that contain $v$, and let $v_i$ be the unique vertex in $\inn{(B_i)}$, for $i \in \{1,2\}$. 
Then let $U = \{s\} \cup Q_0' \cup \{v_i : v \in Q_0', i \in \{1,2\}\}$, and observe that $G[U]$ is a tree with $3|Q_0'|$ edges. It follows that $G[U] \in \Pi$ and $\ex(G[U]) = \frac{3|Q_0'|}{4}$. By Lemma \ref{lem:half}, $\ex(G)\geq\frac{3|Q_0'| -c}{4}$, where $c$ is the number of components of $G-U$.

Consider the components of $G-U$. Each component either contains a block in $\mathcal{B}_1\cup\mathcal{B}_{\geq 3}$ or it is part of a block path of $G-S$ containing two vertices from $Q_0'$: by Corollary \ref{cor:nonpathblocksbound} there are at most $4((\frac{16}{1-\lambda}+\frac{2}{\mak})k-2)\frac{6k}{1-\lambda}+\frac{k}{\mak} = ((32+\frac{2}{\mak})k-2)48k+\frac{4k}{\mak}$ components of the first kind, while there are at most $|Q_0'|$  of the second kind.

Thus, if $2|Q_0'| - ((32+\frac{2}{\mak})k-2)48k-\frac{4k}{\mak}\ge 4k$ then we have a {\sc Yes}-instance; otherwise $|Q_0\cap N(s)|\le 2|Q_0'|\le ((32+\frac{2}{\mak})k-2)48k-\frac{4k}{\mak} + 4k$. 
\end{proof}

\begin{krule}\label{krule:rulenotrianglepath}
 Let $B_1,B_2\in\mathcal{B}_2$ be such that $V(B_1)\cap V(B_2)=\{v\}$, $B_1=\tor=B_2$, $\{v\}\cap N(S)=\emptyset$ and $\inn{(B_i)}\cap N(S)=\emptyset$ for $i=1,2$. Let $\{w_i\}=\inn{(B_i)}$ and $\{u_i\}=V(B_i)\setminus\{v,w_i\}$ for $i=1,2$. If $G-\{v\}$ is disconnected, delete $v,w_1,w_2$, identify $u_1$ and $u_2$ and set $k'=k$. Otherwise, delete $v,w_1,w_2$ and set $k'=k-\frac{1}{4}$.
\end{krule}

\begin{lemma}\label{lem:rulenotrianglepathvalid}
 Let $\Pi$ be a strongly $\lambda$-extendible property on oriented graphs, with $\lambda=\frac{1}{2}$, and suppose $\Pi$ is hereditary. If $\tor\notin\Pi$, then Rule \ref{krule:rulenotrianglepath} is valid.
\end{lemma}
\begin{proof}
 Let $G'$ be the graph which is obtained after an application of the rule. If $G-\{v\}$ is disconnected, let $G''$ be the graph obtained from $G$ deleting $v,w_1$ and $w_2$ and without identifying any vertices. Then, note that $G''$ has two connected components, one containing $u_1$ and the other containing $u_2$: hence, $G'$ is connected. Additionally, $\ex(G')=\ex(G'')$.

 Let $\widetilde{G}=G''$ if $G-\{v\}$ is disconnected and $\widetilde{G}=G'$ otherwise. We have to show that $\ex(\widetilde{G})=\ex(G)$ if $G-\{v\}$ is disconnected and $\ex(\widetilde{G})=\ex(G)-\frac{1}{4}$ otherwise. In order to do this, we will show that for any maximal $\Pi$-subgraph $\widetilde{H}$ of $\widetilde{G}$ there exists a $\Pi$-subgraph $H$ of $G$ such that $H[V(\widetilde{G})]=\widetilde{H}$ and $|E(H[V(B_1)\cup V(B_2)])|=\pt({G[V(B_1)\cup V(B_2)]})$. Then the result follows because if $G-\{v\}$ is disconnected, then $\pt(G)=\pt(G-\{v,w_1,w_2\})+\pt(G[V(B_1)\cup V(B_2)])$, and if $G-\{v\}$ is connected, then $\pt(G)=\pt(G-\{v,w_1,w_2\})+\pt(G[V(B_1)\cup V(B_2)])-\frac{1}{4}$.
 
 Let $\widetilde{H}$ be any maximal $\Pi$-subgraph of $\widetilde{G}$. Note that by block additivity and inclusiveness, $G[v,w_1,w_2]\in\Pi$. Then, by the strong $\lambda$-subgraph extension property there exists a $\Pi$-subgraph $H$ of $G$ which contains $\widetilde{H}$, $G[v,w_1,w_2]$ and at least half of the edges between them. Note that these edges are exactly four: $vu_1$, $w_1u_1$, $vu_2$ and $w_2u_2$. If $vu_1$ and $w_1u_1$ are in $E(H)$, then since $\Pi$ is hereditary it holds that $\tor\in\Pi$, which is a contradiction. Similarly if $vu_2$ and $w_2u_2$ are in $E(H)$. This means that exactly two edges among them are in $E(H)$, that is that $|E(H[V(B_1)\cup V(B_2)])|=4=\pt({G[V(B_1)\cup V(B_2)]})$. 
\end{proof}

\begin{theorem}\label{thm:oriented}
Let $\Pi$ be a strongly $\lambda$-extendible property on oriented graphs, with $\lambda=\frac{1}{2}$, and suppose $\Pi$ is hereditary. Then \APT{} has a kernel with \(\cO(k^{2})\) if $K_3\in\Pi$ and has a kernel with \(\cO(k^{3})\) vertices otherwise.
\end{theorem}
\begin{proof}
 If $\tor\in\Pi$, by Lemma \ref{lem:torthenall} $\tuo\in\Pi$. This means that $K_3\in\Pi$ and, by Theorem \ref{thm:easycase}, \APT{} has a kernel with \(\cO(k^{2})\) vertices. On the other hand, if $\tor\notin\Pi$ and $\tuo\notin\Pi$, then by Theorem \ref{theo:MaxCutEverywhere} $\Pi$ is equal to Max Cut and by Theorem \ref{thm:maxCutKernel} it admits a kernel with \(\cO(k^{3})\) vertices.

 Lastly, suppose $\tuo\in\Pi$ and $\tor\notin\Pi$. By Lemma \ref{lem:tuothendiverges}, $\LE$ diverges on cliques. Let $(G,k)$ be an instance of \APT{} reduced by Reduction Rule \ref{krule:zeroexcess} and \ref{krule:rulenotrianglepath} (note that in this case Rule \ref{krule:rulenotrianglepath} is valid by Lemma \ref{lem:rulenotrianglepathvalid}).

By Corollary \ref{cor:nonpathblocksbound}, we may assume $|\mathcal{B}_0|+|\mathcal{B}_1|+|\mathcal{B}_{\geq 3}|\leq
4(((\frac{16}{1-\lambda}+\frac{2}{\mak})k-2)\frac{6k}{1-\lambda}+\frac{k}{\mak})$.

We now need to consider different types of blocks in $\mathcal{B}_2$ separately. Let $\mathcal{B}_2^+$ be the blocks in $\mathcal{B}_2$ with positive excess.
By Lemma \ref{lem:positiveexcessblock}, we may assume the number of such blocks is at most $(32+\frac{2}{\mak})k-1)\frac{12k}{\mak} + \frac{k}{(\mak)^2} + \frac{k-1}{\mak}$.

Let $\mathcal{B}_2'$ be the blocks in  $\mathcal{B}_2 \setminus \mathcal{B}_2^+$ which have an interior vertex in $N(S)$. By Corollary \ref{cor:sneighborbound}, we may assume the number of such blocks is at most $((32+\frac{2}{\mak})k-2)12k$.

Let $\mathcal{B}_2''$ be the blocks in $\mathcal{B}_2 \setminus (\mathcal{B}_2^+ \cup \mathcal{B}_2')$ which contain a vertex in $Q \cap N(S)$. Observe that these blocks must either contain a vertex of $Q_0\cap N(S)$ or be adjacent to a block in $\mathcal{B}_1,\mathcal{B}_{\geq 3}, \mathcal{B}_2^+$ or $\mathcal{B}_2'$. Furthermore they must be in block paths between such blocks, from which it follows that  $|\mathcal{B}_2''| \le 2(|\mathcal{B}_1|+|\mathcal{B}_{\geq 3}|+|\mathcal{B}_2^+|+|\mathcal{B}_2'| + |Q_0 \cap N(S)|)$.

%
Finally let $\mathcal{B}_2''' = \mathcal{B}_2 \setminus (\mathcal{B}_2^+ \cup \mathcal{B}_2' \cup \mathcal{B}_2'')$. These are just the blocks in $\mathcal{B}_2$ with excess $0$ which contain no neighbors of $S$.
By Reduction Rule \ref{krule:rulenotrianglepath}, no two such blocks can be adjacent. Therefore every block in $\mathcal{B}_2'''$ is adjacent to two blocks from $\mathcal{B}_1,\mathcal{B}_{\geq 3}, \mathcal{B}_2^+, \mathcal{B}_2'$ or $\mathcal{B}_2''$. It follows that $|\mathcal{B}_2'''| \le |\mathcal{B}_1|+|\mathcal{B}_{\geq 3}|+|\mathcal{B}_2^+|+|\mathcal{B}_2'| + |\mathcal{B}_2''|$.

Hence, we may conclude that $|\mathcal{B}_0|+|\mathcal{B}_1|+|\mathcal{B}_{\geq 3}|+|\mathcal{B}_2^+|\le c_1k^2$ for some constant $c_1$ depending only on $\Pi$. Furthermore, note that by Lemma \ref{theo:neighborsbound} and the fact that $|S| \le 12k$, we may assume that $|Q_0 \cap N(S)| \le (((32+\frac{2}{\mak})k-2)48k-\frac{4k}{\mak} + 4k)12k$. Then we may conclude from the above that $|\mathcal{B}_2'| + |\mathcal{B}_2''|+|\mathcal{B}_2'''| \le c_2k^3$ for some constant $c_2$ depending only on $\Pi$.

Therefore the total number of blocks in $G-S$ is at most $c_1k^2 + c_2k^3$. It follows that $|Q|$, the number of cut vertices of $G-S$, is at most $c_1k^2 + c_2k^3$.

By Lemma \ref{lem:int}, we may assume that the number of internal vertices for any block is at most $c_3k$, for some constant $c_3$ depending only on $\Pi$. It follows that the number of vertices in blocks from $\mathcal{B}_0,\mathcal{B}_1,\mathcal{B}_{\geq 3}$ or $\mathcal{B}_2^+$ is at most $c_1c_3k^3 + c_1k^2 + c_2k^3$.
To bound the number of vertices in blocks from $\mathcal{B}_2'\cup\mathcal{B}_2''\cup\mathcal{B}_2'''$, note that each of these blocks contains at most $3$ vertices, by Lemma \ref{lem:tuothenorientedtight} and the fact that these blocks have excess $0$ by definition.
Therefore the number of vertices in blocks from $\mathcal{B}_2'\cup\mathcal{B}_2''\cup\mathcal{B}_2'''$ is at most $3c_2k^3$.
Finally, recalling that $|S|\leq 12k$, we have that the number of vertices in $G$ is \(\cO(k^{3})\). 
\end{proof}

Putting together Theorem \ref{thm:easycase}, Theorem \ref{thm:simple}, and Theorem \ref{thm:oriented}, we get our main result, Theorem \ref{thm:main}.

%
%
%
%

\section{Conclusion}
We have succeeded in showing that \APT{} has a polynomial kernel for nearly all strongly $\lambda$-extendible $\Pi$. The only cases in which the polynomial kernel question remains open are those in which $\lambda = \frac{1}{2}$ and either $\Pi$ is not hereditary, or membership in $\Pi$ depends on the labellings on edges.
For the cases when $\lambda \neq \frac{1}{2}$ or $\Pi$ contains all triangles, we could show the existence of a kernel with \(\cO(k^{2})\) vertices. It would be desirable to show a \(\cO(k^{2})\) kernel in all cases.

The bound of \PandT{} extends to edge-weighted graphs - for any strongly $\lambda$-extendible property $\Pi$ and any connected graph $G$ with weight function $w: E(G) \rightarrow \mathbb{R}^+$, there exists a subgraph $H$ of $G$ such that $H \in \Pi$ and $w(H) \ge \lambda w(G) + \frac{(1-\lambda)w(T)}{2}$, where $T$ is a minimum weight spanning tree of $G$. 
The natural question following from our results is whether the weighted version of \APT{} affords a polynomial kernel.

\bibliographystyle{plain}
{\bibliography{references}}

\begin{thebibliography}{10}

\bibitem{Alon:2010vp}
Noga Alon, Gregory Gutin, Eun~Jung Kim, Stefan Szeider, and Anders Yeo.
\newblock Solving max-r-sat above a tight lower bound.
\newblock In {\em SODA 2010}, pages 511--517. ACM-SIAM, 2010.

\bibitem{Bodlaender09}
Hans~L. Bodlaender.
\newblock Kernelization: New upper and lower bound techniques.
\newblock In {\em Proceedings of the 4th Workshop on Parameterized and Exact
  Computation (IWPEC 2009)}, volume 5917 of {\em Lecture Notes in Computer
  Science}, pages 17--37, 2009.

\bibitem{BodlaenderDFH09}
Hans~L. Bodlaender, Rodney~G. Downey, Michael~R. Fellows, and Danny Hermelin.
\newblock On problems without polynomial kernels.
\newblock {\em J. Comput. Syst. Sci.}, 75(8):423--434, 2009.

\bibitem{H.Bodlaender:2009ng}
Hans~L. Bodlaender, Fedor~V. Fomin, Daniel Lokshtanov, Eelko Penninkx, Saket
  Saurabh, and Dimitrios~M. Thilikos.
\newblock ({M}eta) {K}ernelization.
\newblock In {\em FOCS 2009}, pages 629--638. IEEE, 2009.

\bibitem{CrowstonGJM12}
Robert Crowston, Gregory Gutin, Mark Jones, and Gabriele Muciaccia.
\newblock Maximum balanced subgraph problem parameterized above lower bound.
\newblock {\em CoRR}, abs/1212.6848, 2012.

\bibitem{CrowstonJonesMnich2011}
Robert Crowston, Mark Jones, and Matthias Mnich.
\newblock Max-cut parameterized above the {Edwards-Erd\H{o}s} bound.
\newblock {\em CoRR}, abs/1112.3506, 2011.
\newblock http://arxiv.org/abs/1112.3506.

\bibitem{CrowstonJonesMnich2012}
Robert Crowston, Mark Jones, and Matthias Mnich.
\newblock Max-cut parameterized above the {Edwards-Erd\H{o}s} bound.
\newblock In Artur Czumaj, Kurt Mehlhorn, Andrew~M. Pitts, and Roger
  Wattenhofer, editors, {\em Automata, Languages, and Programming - 39th
  International Colloquium, ICALP 2012, Warwick, UK, July 9-13, 2012,
  Proceedings, Part I}, volume 7391 of {\em Lecture Notes in Computer Science},
  pages 242--253. Springer, 2012.

\bibitem{Dell:2010sh}
Holger Dell and Dieter van Melkebeek.
\newblock Satisfiability allows no nontrivial sparsification unless the
  polynomial-time hierarchy collapses.
\newblock In {\em STOC 2010}, 2010.

\bibitem{Diestel}
Reinhard Diestel.
\newblock {\em Graph Theory}, volume 173 of {\em Graduate Texts in
  Mathematics}.
\newblock Springer, 2010.

\bibitem{Edwards1973}
C.~S Edwards.
\newblock Some extremal properties of bipartite subgraphs.
\newblock {\em Canadian Journal of Mathematics}, 25:475--483, 1973.

\bibitem{Edwards1975}
C.~S Edwards.
\newblock An improved lower bound for the number of edges in a largest
  bipartite subgraph.
\newblock In {\em Recent Advances in Graph Theory}, pages 167--181. 1975.

\bibitem{FominLMS12}
Fedor~V. Fomin, Daniel Lokshtanov, Neeldhara Misra, and Saket Saurabh.
\newblock Planar {F}-deletion: Approximation, {K}ernelization and optimal {FPT}
  algorithms.
\newblock In {\em FOCS}, pages 470--479, 2012.

\bibitem{FominLST10}
Fedor~V. Fomin, Daniel Lokshtanov, Saket Saurabh, and Dimitrios~M. Thilikos.
\newblock Bidimensionality and kernels.
\newblock In {\em SODA}, pages 503--510, 2010.

\bibitem{FS08}
Lance Fortnow and Rahul Santhanam.
\newblock Infeasibility of instance compression and succinct {PCPs} for~{NP}.
\newblock In {\em STOC 08}, pages 133--142. ACM, 2008.

\bibitem{GN07SIGACT}
Jiong Guo and Rolf Niedermeier.
\newblock Invitation to data reduction and problem kernelization.
\newblock {\em SIGACT news}, 38(1):31--45, 2007.

\bibitem{FlumGroheBook}
{J}{\"o}rg {F}lum and {M}artin {G}rohe.
\newblock {\em {P}arameterized {C}omplexity {T}heory}.
\newblock Springer -Verlag, 2006.

\bibitem{KratschW12Soda}
Stefan Kratsch and Magnus Wahlstr{\"o}m.
\newblock Compression via matroids: a randomized polynomial kernel for odd
  cycle transversal.
\newblock In {\em SODA}, pages 94--103, 2012.

\bibitem{KratschW12}
Stefan Kratsch and Magnus Wahlstr{\"o}m.
\newblock Representative sets and irrelevant vertices: New tools for
  kernelization.
\newblock In {\em Proceedings of the 53rd Annual Symposium on Foundations of
  Computer Science (FOCS 2012)}, pages 450--459. IEEE, 2012.

\bibitem{MahajanRaman1999}
Meena Mahajan and Venkatesh Raman.
\newblock Parameterizing above guaranteed values: {MaxSat} and {MaxCut}.
\newblock {\em Journal of Algorithms}, 31(2):335--354, 1999.

\bibitem{MnichPhilipSaurabhSuchy2012}
Matthias Mnich, Geevarghese Philip, Saket Saurabh, and Ondrej Such{\'y}.
\newblock Beyond max-cut: $\lambda$-extendible properties parameterized above
  the {P}oljak-{T}urz{\'i}k bound.
\newblock In Deepak D'Souza, Telikepalli Kavitha, and Jaikumar Radhakrishnan,
  editors, {\em IARCS Annual Conference on Foundations of Software Technology
  and Theoretical Computer Science, FSTTCS 2012, December 15-17, 2012,
  Hyderabad, India}, volume~18 of {\em Leibniz International Proceedings in
  Informatics (LIPIcs)}, pages 412--423. Schloss Dagstuhl - Leibniz-Zentrum
  fuer Informatik, 2012.

\bibitem{MnichPhilipSaurabhSuchy2012arXiv}
Matthias Mnich, Geevarghese Philip, Saket Saurabh, and Ondrej Such{\'y}.
\newblock Beyond max-cut: $\lambda$-extendible properties parameterized above
  the {P}oljak-{T}urz{\'i}k bound.
\newblock {\em CoRR}, abs/1207.5696, 2012.
\newblock http://arxiv.org/abs/1207.5696.

\bibitem{Niedermeierbook06}
Rolf Niedermeier.
\newblock {\em Invitation to {F}ixed-{P}arameter {A}lgorithms}, volume~31 of
  {\em Oxford Lecture Series in Mathematics and its Applications}.
\newblock Oxford University Press, Oxford, 2006.

\bibitem{PoljakTurzik1986}
Svatopluk Poljak and Daniel Turz{\'i}k.
\newblock A polynomial time heuristic for certain subgraph optimization
  problems with guaranteed worst case bound.
\newblock {\em Discrete Mathematics}, 58(1):99--104, 1986.

\bibitem{Karp1972}
{R}ichard {M}.~{K}arp.
\newblock {R}educibility among combinatorial problems.
\newblock In {\em Complexity of Computer Communications}, pages 85--103, 1972.

\end{thebibliography}

%

\end{document}